\documentclass[12pt]{article}
\usepackage[utf8]{inputenc}
\usepackage{geometry}
\usepackage{color}
\usepackage{amsthm}
\usepackage{amsmath}
\usepackage{amssymb}
\usepackage{esint}
\usepackage[authoryear]{natbib}

\makeatletter
  \theoremstyle{definition}
  \newtheorem{defn}{\protect\definitionname}
  \theoremstyle{remark}
  \newtheorem{rem}{\protect\remarkname}
\theoremstyle{plain}
\newtheorem{thm}{\protect\theoremname}
  \theoremstyle{plain}
  \newtheorem{cor}{\protect\corollaryname}
 \theoremstyle{definition}
  \newtheorem{example}{\protect\examplename}
  \theoremstyle{plain}
  \newtheorem{lem}{\protect\lemmaname}

\@ifundefined{date}{}{\date{}}

\title{On Optimal Reinsurace Policy}

\author{Hirbod Assa \footnote{Mailing address: Institute for Financial and Actuarial Mathematics. University of Liverpool. UK. Email: assa@liverpool.ac.uk}\\
University of Liverpool \and
}

  \providecommand{\corollaryname}{Corollary}
  
  \providecommand{\remarkname}{Remark}
\providecommand{\theoremname}{Theorem}

  \providecommand{\corollaryname}{Corollary}
  \providecommand{\definitionname}{Definition}
  \providecommand{\examplename}{Example}
  \providecommand{\lemmaname}{Lemma}
  \providecommand{\remarkname}{Remark}
\providecommand{\theoremname}{Theorem}

\makeatother

  \providecommand{\definitionname}{Definition}
  \providecommand{\examplename}{Example}
  \providecommand{\lemmaname}{Lemma}
  \providecommand{\remarkname}{Remark}
\providecommand{\corollaryname}{Corollary}
\providecommand{\theoremname}{Theorem}

\begin{document}

\title{Optimal risk allocation in a market with non-convex preferences}
\maketitle
\begin{abstract}
\noindent The aims of this study are twofold. First, we consider an
optimal risk allocation problem with non-convex preferences. By establishing
an infimal representation for distortion risk measures, we give some
necessary and sufficient conditions for the existence of optimal and
asymptotic optimal allocations. We will show that, similar to a market
with convex preferences, in a non-convex framework with distortion
risk measures the boundedness of the optimal risk allocation problem
depends only on the preferences. Second, we consider the same optimal
allocation problem by adding a further assumption that allocations
are co-monotone. We characterize the co-monotone optimal risk allocations
within which we prove the “marginal risk allocations” take only the
values zero or one. Remarkably, we can separate the role of the market
preferences and the total risk in our representation.
\end{abstract}

\section{Introduction}

There is considerable interest in the problem of optimal risk allocation,
as it is at the heart of many financial and insurance applications.
Optimal risk sharing, optimal capital allocation, theory of market
equilibrium, optimal reinsurance design and optimal risk exchange
are only a few examples. This problem dates back to the 50s and 60s
when \citet{Allais:1953}, \citet{Arrow:1964}, \citet{Sharpe:1964},
\citet{Borch:1960}, \citet{Mossin:1966} and many others studied
the optimal risk allocations for different economic problems. Thereafter,
researchers started to elaborate further on the aspects of this problem
for a variety of assumptions. By development of risk measures and
their applications in finance and insurance, the problem of optimal
risk allocation has been revisited by using coherent risk measures
of \citet{Artzner/Delbaen/Eber/Heath:1999}, convex risk measures
of \citet{Follmer/Schied:2002(1)} and deviation measures of risk
of \citet{Rockafellar/Uryasev/Zabarankin:2006}. The first attempt
to study the problem in a setting with coherent risk measures was
\citet{Heath/Ku:2004}, where the authors established a necessary
and sufficient condition for the existence of a Pareto optimal allocation.
\citet{Barrieu/ElKaroui:2004} considered a risk sharing problem in
a dynamic setup, whereas \citet{Jouini/Schachermayer/Touzi:2008}
considered a static framework with law-invariant convex risk measures.
\citet{Filipovic/Kupper:2008} looked at the optimal risk allocation
problem from a pricing point of view, while \citet{Filipovic/Kupper:2008(1)}
considered it for optimal capital allocations. \citet{Acciaio:2007}
studied a sharing pooled risk problem with non-necessarily monotone
monetary utilities. While there is extensive research on the problem
of optimal risk allocation with convex preferences, studies using
non-convex framework have been relatively scarce, whereas in many
applications preferences are not convex, and the results of the existing
settings cannot be applied to them. This is mainly due to the lack
of appropriate mathematical techniques to study models with the non-convex
preferences. 

In this paper, by establishing an infimal representation for distortion
risk measures, we find a new way to study the optimal risk allocation
problem with non-convex preferences. We prove that the boundedness
of the optimal risk allocation problem is independent of the total
risk and only depends on the market preferences. The approach we have
chosen is a finance oriented approach which gives rise to the definition
of generalized stochastic discount factors for non-convex preferences
(see Remark 2 below). Our results generalize results of \citet{Jouini/Schachermayer/Touzi:2008}
, \citet{Filipovic/Kupper:2008} and \citet{Filipovic/Kupper:2008(1)}
towards a new direction by using non-convex risk measures. This constitutes
the first part of the paper. In the second part, with an extra assumption
that the risk allocations are co-monotone, we characterize the optimal
risk allocations in the same market. This assumption can be interpreted
as mutualization of risks, which is closely related to the moral hazard
risk%
\footnote{ In order to avoid the moral hazard risk, allocations have to be increasing
in terms of market risks.%
}. Interestingly, we see that the optimal risk allocations in a setting
with distortion risk measures are in a perfect accordance with this
assumption. It is shown in \citet{Filipovic/Svindland:2008(1)} that
the solutions to a general market risk allocation problem with convex
distortion risk measures are co-monotone, and therefore, rule out
the risk of moral hazard. However, we will see within an example that
this no longer holds true when agents use non-convex distortion risk
measures. That is why we have to assume that the allocations are co-monotone.
In order to characterize the co-monotone optimal solutions, we introduce
the “marginal risk allocations”. A marginal risk allocation is the
marginal rate of changes in the value of a contract when we marginally
change the value of the total risk. It is shown that, in a market
with co-monotone risk allocations, the marginal risk allocations take
only the values zero or one. This way, we can remarkably separate
the role of the market preferences and the total risk in the optimal
risk allocations. Our results find a new characterization of the optimal
allocations in \citet{Chateauneuf/Dana/Tallon:2000} enabling us for
more precise interpretation of the optimal allocations and also finding
further applications in other fields such as the optimal re-insurance
design. This paper generalize the literature of optimal re-insurance
design in two directions. First, we use a larger family of (non-convex)
risk measures and premiums and second, we increase the number of players
from two to $n$ (e.g. see, \citet{Cai/Tan/Weng/Zhang:2008}, \citet{Cheung:2010},\citet{Chi:2012},
\citet{Chi:2012(2)}, \citet{Chi/Tan:2013} , \citet{Cheung/Sung/Yam/Yung:2014}
and \citet{Assa:2015}). 

The rest of the paper is organized as follows: in Section 2 we introduce
the needed notions and notations, and recall some facts from convex
analysis. In Section 3, first, we set up the main problem, second,
we discuss some necessary and sufficient conditions for the existence
of general solutions, and, third, we characterize the co-monotone
optimal solutions.

\section{Preliminaries and Notations}

Throughout the paper, we will fix a probability space $\left(\Omega,\mathcal{F},P\right)$,
where $\mathcal{F}$ is a $\sigma$-algebra and $P$ is a probability
measure on $\mathcal{F}$. Let $p,q\in[1,\infty]$ be two numbers
such that $1/p+1/q=1$. For $p\not=\infty$, $L^{p}$ denotes the
space of real-valued random variables $X$ on $\Omega$ such that
$E\left(\left|X\right|^{p}\right)<\infty$, where $E$ represents
the mathematical expectation. Recall that according to the Riesz Representation
Theorem, $L^{q}$ is the dual space of $L^{p}$ when $p\neq\infty$.
We endow the space $L^{p}$ with two topologies, first the norm topology
induced by $\Vert X\Vert_{p}=E(\vert X\vert^{p})^{\frac{1}{p}}$,
and second the weak topology, induced by $L^{q}$ i.e. the coarsest
topology in which all members of $L^{q}$ are continuous. As usual
the latter topology is denoted by $\sigma(L^{p},L^{q})$. 

In this paper we consider that $L^{p}$ represents the space of all
loss variables%
\footnote{Unlike finance literature which consider profit variable, we found
the loss variable more convenient to deal with. %
}. We only have two periods of time $0$ and $T$, representing the
beginning of the year when a contract is written, and the end the
year when liabilities are settled, respectively. Every random variable
represents losses at time $T$. Whenever we talk about risk or premium
we mean the present value of the loss and the premium at time $T=0$.

\subsection{Distortion Risk Measures}

\textcolor{black}{Let $\Phi:[0,1]\to[0,1]$ be a non-decreasing and
c\'adl\'ag function such that $\Phi(0)=1-\Phi(1)=0$. $\Phi$ can
introduce a measure on $\left[0,1\right]$ whose values on the intervals
are given as $m_{\Phi}\left[a,b\right)=\Phi\left(b\right)-\Phi\left(a\right)$
and $m_{\Phi}\left(b\right)=1-\lim_{a\uparrow1}\Phi\left(a\right)$.
Introduce the set $\mathcal{D}_{\Phi}$ as follows}

\textcolor{black}{
\begin{equation}
\mathcal{D}_{\Phi}=\left\{ X\in L^{0}\mid\int_{0}^{1}\mathrm{\mathrm{VaR}}_{t}(X)d\Phi(t)\in\mathbb{R}\right\} ,
\end{equation}
where the integral above is the Lebesgue integral and 
\[
\mathrm{\mathrm{VaR}}_{\alpha}(X)=\inf\{x\in\mathbb{R}|P(X>x)\le1-\alpha\},\alpha\in[0,1].
\]
}
\begin{defn}
\textcolor{black}{A distortion risk measure $\varrho_{\Phi}$ (or
simply $\varrho$) is a mapping from $\mathcal{D}_{\Phi}$ to $\mathbb{R}$
defined as }
\begin{equation}
\varrho_{\Phi}(X)=\int_{0}^{1}\mathrm{\mathrm{VaR}}_{t}(X)d\Phi(t),\label{spectral_risk}
\end{equation}
\textcolor{black}{If we let $g(x):=1-\Phi(1-x)$ one can see that
\begin{equation}
\varrho_{\Phi}(X)=\int\limits _{-\infty}^{0}\left(g(S_{X}(t))-1\right)dt+\int\limits _{0}^{\infty}g(S_{X}(t))dt,\label{eq:distortion}
\end{equation}
where $S_{X}=1-F_{X}$ is the survival function associated with $X$.
Note that we can associate $\varrho$ with $\Phi$ by using the notation
$\Phi_{\varrho}$. This is a Choquet integral representation of the
risk measure. In the literature, $g$ is known as the distortion function.}
A popular example is Value at Risk (VaR), whose distortion function
is given by $g(t)=1_{[1-\alpha,1]}(t)$ for a confidence level $1-\alpha$.
It can also explicitly be given as 
\[
\mathrm{\mathrm{VaR}}_{\alpha}(X)=\inf\{x\in\mathbb{R}|P(X>x)\le1-\alpha\}.
\]
Another example of a distortion risk measure is Conditional Value
at Risk (CVaR), when $\Phi(t)=\frac{t-\alpha}{1-\alpha}1_{[\alpha,1]}(t)$
and can be represented in terms of VaR 
\begin{equation}
\mathrm{CVaR}_{\alpha}(x)=\frac{1}{1-\alpha}\int_{\alpha}^{1}\mathrm{\mathrm{VaR}}_{t}(X)dt.\label{eq:CVaR}
\end{equation}
The family of spectral risk measures which was introduced first in
\citet{Acerbi:2002}, is a distortion risk when $\Phi$ is convex.\end{defn}
\begin{rem}
One can readily see that $\varrho_{\Phi}$ is law invariant, i.e.,
if $X$ and $X^{\prime}$ are identically distributed, then we have
$\varrho_{\Phi}(X)=\varrho_{\Phi}(X^{\prime})$. Indeed, it can be
shown that all law-invariant co-monotone additive coherent risk measures
can be represented as \eqref{spectral_risk}; see \citet{Kusuoka:2001}.
A risk measure in the form \eqref{spectral_risk} is important from
different perspectives. First of all, it makes a link between the
risk measure theory and the behavioral finance as the form \eqref{spectral_risk}
is a particular form of distortion utility. Second, \eqref{spectral_risk}
contains a family of risk measures which are statistically robust.
In \citet{Cont/Deguest/Scandolo:2010} it is shown that a risk measure
$\varrho(x)=\int_{0}^{1}\mathrm{\mathrm{VaR}}_{t}(x)d\Phi(t)$ is
robust if and only if the support of $\varphi=\frac{d\Phi(t)}{dt}$%
\footnote{$\varphi$ is a general derivative of $\Phi$.%
} is away from zero or one. For example Value at Risk is a risk measure
with this property. distortion utilities have become increasingly
important in the literature of decision making since they take into
account some known behavioral paradoxes such as the Allais paradox
under risk and the Ellsberg paradox under uncertainty. \citet{Schmeidler:1989}
(under uncertainty) and \citet{Quiggin:1982} and \citet{Yaari:1984},
\citet{Yaari:1986} (under risk) show by assuming co-monotone independence,
preferences are according to utilities which admit a distortion integral
representation. It is worth mentioning that, distortion integrals
have become very popular in the literature of insurance because they
are the natural extensions of important insurance risk premiums such
as Proportional Hazards Premium Principle, Wang’s Premium Principle
and Net Premium Principle (see \citet{Wang/Young/Panjer:2006} and
\citet{Young:2006}).
\end{rem}
Finally, we have the definition of a coherent risk measure
\begin{defn}
A coherent risk measure $\varrho$ is a lower semi-continuous%
\footnote{The risk measure in general does not need to be lower semi-continuous
in $L^{\infty}$, however we add it to be consistent with $L^{p},p\not=\infty$.%
}(see below for definition of lower semi-continuous) mapping from $L^{p}$
to $\mathbb{{R}}\cup\left\{ +\infty\right\} $ such that\end{defn}
\begin{enumerate}
\item $\varrho(\lambda X)=\lambda\varrho(X)$, for all $\lambda>0$ and
$X\in L^{p}$;
\item $\varrho(X+c)=\varrho(X)+c$, for all $X\in L^{p}$ and $c\in\mathbb{R}$;
\item $\varrho(X)\leq\varrho(Y)$, for all $X,Y\in L^{p}$ and $X\le Y$;
\item $\varrho(X+Y)\leq\varrho(X)+\varrho(Y),\forall X,Y\in L^{p};$
\end{enumerate}
As one can see, a coherent risk measure is positive homogeneous. As
we will see in the next section, there is a closed and convex subset
$\Delta_{\varrho}\subseteq L^{q}$, such that $\varrho(X)=\sup{}_{Y\in\Delta_{\varrho}}E(YX)$.
One can show that for any $Y\in\Delta_{\varrho}$, we have $E(Y)=1$
and $Y\ge0$.

\subsection{Some Facts from Convex Analysis}

Here we recall some relevant discussions from the convex analysis.
Recalling from the convex analysis, for any convex function $\phi$,
the domain of $\phi$ denoted by $\mathrm{dom(\phi)}$ is equal to
$\left\{ X\in L^{p}|\phi(X)<\infty\right\} $, and the dual of $\phi$,
denoted by $\phi^{*}$, is defined as $\phi^{*}(Y)=\sup_{X\in L^{p}}E(XY)-\phi(X)$.
A convex function is called lower-semi-continuous iff $\phi=\phi^{**}$.
In this paper, we assume all convex functions are lower semi continuous.
For a convex set $C\subseteq L^{p},$ the indicator function of $C$
is denoted by $\chi_{C}$ and is equal to $0$ if $X\in C$, and $+\infty$,
otherwise. One can incorporate any type of convex restriction by using
an appropriate indicator function. Let $C$ be a closed and convex
set representing a convex restriction on $\phi$. By introducing $\phi^{C}=\phi+\chi_{C}$
we incorporate the restriction $C$. Note that $\phi^{C}$ is a convex
function.

For any positive homogeneous convex function $\phi$ let

\[
\Delta_{\phi}=\left\{ Y\in L^{q}|E(YX)\le\phi(X),\forall X\in L^{p}\right\} .
\]
It is easy to see that $\phi^{*}=\chi_{\Delta_{\phi}}$. Therefore,
any positive homogeneous function $\phi$ can be represented as $\phi(X)=\sup\limits _{Y\in\Delta_{\phi}}E(YX)$.
By using this and that $\phi=\phi^{**}$, one can easily see that
for any convex set $C$, $\chi_{C}^{*}(Y)=\sup\limits _{X\in C}E(YX)$.
As a result, if $C$ is a convex cone and $\phi$ is a positive homogeneous
convex function then $\phi^{C}(X)=\sup\limits _{Y\in\Delta_{\phi}+C^{\perp}}E(YX)$,
where $C^{\perp}=\left\{ Y\in L^{q};E(YX)\le0,\forall X\in C\right\} $
(or $\Delta_{\phi^{C}}=\Delta+C^{\perp}$). A particular interesting
example is $C=L_{+}^{p}$ when $C^{\perp}=L_{-}^{q}$.

For a set of convex functions $\phi_{1},...,\phi_{n}$ their infimal
convolution is defined as 
\begin{equation}
\phi_{1}\square...\square\phi_{n}(X)=\inf_{X_{1}+...+X_{n}=X}\phi_{1}(X_{1})+...+\phi_{n}(X_{n}).\label{eq:conolution_general}
\end{equation}
In \citet{Rockafellar/book:1997} Theorems 5.4 and 16.4 it is shown
that $(\phi_{1}\square...\square\phi_{n})^{*}=\phi_{1}^{*}+...+\phi_{n}^{*}$.
By using the arguments above one can easily see that if $\phi_{1},...,\phi_{n}$
are positive homogeneous then $\phi_{1}\square...\square\phi_{n}(X)=\sup\limits _{Y\in\cap_{i}\Delta_{\phi_{i}}}E(YX)$.
As a result
\begin{thm}
\label{Thm:boundedness}The infimum in the infimal convolution is
bounded if and only if $\cap_{i}\Delta_{\phi_{i}}\not=\varnothing$.
\end{thm}
Another classical result is the following
\begin{thm}
\label{Thm:attainable}Assume that $\phi_{1},...,\phi_{n}$ are $n$
positive homogenous convex function. The following two statements
are equivalent \end{thm}
\begin{enumerate}
\item \textit{$(X_{1},...,X_{n})$ is an optimal allocation for $X$ i.e.,
$X_{1}+...+X_{n}=X\text{ and }\phi_{1}(X_{1})+...+\phi_{n}(X_{n})=\phi_{1}\square...\square\phi_{n}(X)$;}
\item \textit{There exists $Y\in L^{q}$ such that $\phi_{i}(X_{i})=E(YX_{i}),i=1,...,n$.}
\end{enumerate}
For a proof one can see \citet{Jouini/Schachermayer/Touzi:2008}.
Let $M_{1},...,M_{n}$ are $n$ convex and closed cones, subsets of
$L^{p},$ representing $n$ constraints that agents $1$ to $n$ face
in the economy. Then by replacing $\phi_{i}$ with $\phi^{M_{i}}$
in the above, we can consider the same setting which also incorporates
the economy constraint in the problem.

And finally the positive infimal convolution is denoted by $\varrho_{1}\boxdot...\boxdot\varrho_{n}$
as is defined as
\[
\varrho_{1}\boxdot...\boxdot\varrho_{n}(X)=\inf_{X_{1}+...+X_{n}=X,X_{i}\ge0,i=1,...,n}\phi_{1}(X_{1})+...+\phi_{n}(X_{n}).
\]

\section{Problem Set-up}

Let us assume there are $n$ different agents in the market whose
preferences are according to $n$ distortion risk measures $\varrho_{1},\dots,\varrho_{n}$.
We denote the associated kernels with $\Phi_{1,}...,\Phi_{n}$. The
risk of the whole market is modeled by a loss variable $X_{0}$. The
set of allocations denoted by $\mathbb{A}$ is defined as follows
\[
\mathbb{A}=\left\{ (X_{1},...,X_{n})\in\left(L^{p}\right)^{n}\vert X_{1}+...+X_{n}=X_{0}\right\} .
\]
An optimal allocation is an allocation which minimizes the aggregate
risk

\begin{equation}
\inf_{X_{1}+...+X_{n}=X_{0}}\varrho_{1}(X_{1})+...+\varrho_{n}(X_{n}),\label{eq:convolution_allocation}
\end{equation}
An asymptotic optimal allocation is a sequence $\left\{ (X_{1}^{m},...,X_{n}^{m})\right\} _{m=1,2,...}\subseteq\mathbb{A}$,
such that 

\begin{equation}
\varrho_{1}(X_{1}^{m})+...+\varrho_{n}(X_{n}^{m})\,\,\underrightarrow{\, m\to\infty\,}\,\,\inf_{X_{1}+...+X_{n}=X_{0}}\varrho_{1}(X_{1})+...+\varrho_{n}(X_{n}).\label{eq:convolution_allocation-asymptotic}
\end{equation}
It is clear that the existence of an asymptotic optimal allocation
is equivalent to the boundedness of \eqref{eq:convolution_allocation}.
For further development of the existing setting we have to consider
a wider problem 
\begin{equation}
\inf\limits _{(X_{1,}...,X_{n})\in\mathbb{A}}\lambda_{1}\varrho_{1}(X_{1})+...+\lambda_{n}\varrho_{n}(X_{n}),\label{eq:main_minimization-1}
\end{equation}
when $(\lambda_{1},...,\lambda_{n})$ is an arbitrary set of positive
numbers. For instance, Pareto allocations in a market whose agent
utilities are $-\varrho_{i},i=1,...,n$, are the solutions to this
problem. We will see that if there is no friction in the market, then
for any set of coherent risk measures $\varrho_{1},...,\varrho_{n}$,
$\lambda_{i}$'s should be equal. On the other hand, in (re-)insurance
studies, one can find a risk sharing problem which has very similar
components; $\varrho_{1}$ is a risk measure, measuring the ceding
company global risk, and $\varrho_{2}$ is a risk premium function,
pricing the reinsurance contracts. In this problem $\lambda_{1}=1$
and $\lambda_{2}=1+\rho$ is a relative safely loading parameter (for
more details see example below).

\subsection{General Solutions}

Our approach in this section is to reduce the risk allocation problem
to an inner problem which can be solved by the existing results in
the literature. Even though the general form of a distortion risk
is not a coherent risk measure, thanks to the following statement
we can use the convex analysis approaches to study \eqref{eq:main_minimization-1}.
\begin{thm}
(Infimal Characterization of distortion Risks)\label{por:1} Let 
\[
\varrho_{\Phi}(X)={\displaystyle \int\nolimits _{0}^{1}\mathrm{\mathrm{VaR}}_{s}(X)d\Phi(s)},
\]
for a non-decreasing function $\Phi$ as in Definition 1. If $\varrho_{\Phi}$
is $L^{p}$ continuous, and $X$ is bounded below, we have the following
equality 
\[
\varrho_{\Phi}(X)=\min\{\varrho(X)|\text{for all l.s.c. coherent risk measures }\varrho\text{ such that }\varrho\geq\varrho_{\Phi}\}.
\]
\end{thm}
\begin{proof}
First, we prove the theorem for $p=\infty$. Consider a sequence of
partitions $\Sigma_{k}=\{\alpha_{0}^{k}=0<\alpha_{1}^{k}<\dots,\alpha_{k}^{k}<\alpha_{k+1}^{k}=1\},k=1,2,...$
of $[0,1]$ such that $\Sigma_{k}\subseteq\Sigma_{k+1}$ and $\mathrm{mesh}(\Sigma_{k})\to0$.
According to Theorem 6.8. in \citet{Delbaen:2000}, for a given $X$
there are coherent risk measures $\varrho_{i}^{k}$, $i=1,\dots,k$
such that $\varrho_{i}^{k}\ge\mathrm{VaR}_{\alpha_{i}^{k}}$ and $\varrho_{i}^{k}(X)=\mathrm{VaR}_{\alpha_{i}^{k}}(X)$.
Define the following mappings on $L^{\infty}$: 
\[
V_{k}(Y)=\sum\limits _{i=0}^{k}(\Phi(\alpha_{i+1}^{k})-\Phi(\alpha_{i}^{k}))\mathrm{VaR}_{\alpha_{i}^{k}}(Y)\text{ and }\varrho_{k}(Y)=\sum\limits _{i=0}^{k}(\Phi(\alpha_{i+1}^{k})-\Phi(\alpha_{i}^{k}))\varrho_{i}^{k}(Y).
\]
Define the coherent risk measure $\varrho$ as 
\[
\varrho(Y)=\lim\sup\limits _{k}\varrho_{k}(Y),\forall Y\in L^{\infty}.
\]
 It is clear that $\varrho$ is $\sigma(L^{\infty},L^{1})$-l.s.c.
Since $\varrho_{k}\ge V_{k}$ and $\varrho_{k}(X)=V_{k}(X)$, by using
the very definition of an integral it turns out that $\varrho\ge\varrho_{\Phi}$
and $\varrho(X)=\varrho_{\Phi}(X)$. 

Now, let us assume that $X\in L^{p}$. Let $\Sigma$ be the set of
all finite sigma algebras on $\Omega$. Recall that $\Sigma$ is a
directed set. For any $\mathcal{G\in}\Sigma$ let $\varrho_{\mathcal{G}}$
be a $\sigma(L^{\infty},L^{1})$-l.s.c coherent risk measure which
dominates $\varrho_{\Phi}$ on $L^{\infty}$ and $\varrho_{\Phi}(E(X|\mathcal{G}))=\varrho_{\mathcal{G}}(E(X|\mathcal{G}))$.
Let $\varrho(Y)=\limsup_{\mathcal{G}}\varrho_{\mathcal{G}}(E(Y|\mathcal{G})),\forall Y\in L^{p}$.
Notice that for any $Z\in L^{1}$, if $Y_{k}\to Y$ weakly in $L^{p}$,
$E(ZE(Y_{k}|\mathcal{G}))=E(E(Z|\mathcal{G})Y_{k})\to E(E(Z|\mathcal{G})Y)=E(ZE(Y|\mathcal{G}))$.
Hence, each function $Y\mapsto\varrho_{\mathcal{G}}(E(Y|\mathcal{G}))$
is $L^{p}$ lower semicontinuous. This implies that $\varrho$ is
also $L^{p}$ lower semicontinuous. On the other hand, for any $Y\in L^{p}$,
the net $\left\{ Y_{\mathcal{G}}=E(Y|\mathcal{G})\right\} _{\mathcal{G}}$
converges in $L^{p}$ and therefore, converges in distribution to
$Y$. This implies that the sequence of functions $\left\{ t\mapsto\mathrm{VaR}_{t}(E(Y|\mathcal{G}))\right\} _{\mathcal{G}}$
converges point-wise to the function $t\mapsto\mathrm{VaR}_{t}(Y)$.
Given that $X$ is bounded below, by using a version of the Fatou
lemma for nets, we have that
\begin{multline*}
\varrho_{\Phi}(Y)=\int_{0}^{1}\mathrm{VaR}_{t}(Y)d\Phi(t)\le\liminf_{\mathcal{G}}\int_{0}^{1}\mathrm{VaR}_{t}(Y_{\mathcal{G}})d\Phi(t)\\
=\liminf_{\mathcal{G}}\varrho_{\Phi}(Y_{\mathcal{G}})\le\liminf_{\mathcal{G}}\varrho_{\mathcal{G}}(E(Y|\mathcal{G}))\le\limsup_{\mathcal{G}}\varrho_{\mathcal{G}}(E(Y|\mathcal{G}))=\varrho(Y)
\end{multline*}
With a similar argument as above, one can show that if instead of
Fatou lemma we use the dominated convergence theorem, and also the
assumption that $\varrho_{\Phi}$ is $L^{p}$ continuous, we have
that $\varrho_{\Phi}(X)=\varrho(X)$.
\end{proof}
The following theorem is almost an immediate result from the previous
theorem and Theorem \ref{Thm:attainable}.
\begin{thm}
Let $\varrho_{1},...,\varrho_{n}$ be $n$ $L^{p}$-continuous distortion
risk measures, $\lambda_{1},...,\lambda_{n}$ ,be $n$ positive numbers
and $M_{1},...,M_{n}$ be $n$ closed convex cones. Let us denote
by $\Lambda_{i}^{M_{i}}$ the set of all functionals $\tilde{\varrho}_{i}^{M_{i}}=\tilde{\varrho_{i}}+\chi_{M_{i}}$,
where $\tilde{\varrho_{i}}$ is a coherent risk measure greater than
or equal to $\varrho_{i}$, $i=1,...,n$. If $X_{0}$ is bounded below,
the following statements for an allocation $(X_{1},...,X_{n})\in\mathbb{A}$
hold\end{thm}
\begin{enumerate}
\item \textit{If $(X_{1},...,X_{n})$ is an optimal allocation for problem
\eqref{eq:convolution_allocation} for all$(\lambda_{1}\tilde{\varrho}^{M_{1}},...,\lambda_{n}\tilde{\varrho}_{n}^{M_{n}})\in\Lambda^{M_{1}}\times...\times\Lambda^{M_{n}}$,
then it is optimal for $(\lambda_{1}\varrho^{M_{1}},...,\lambda_{n}\varrho^{M_{n}})$.}
\item \textit{If $(X_{1},...,X_{n})$ is not optimal for any $(\lambda_{1}\tilde{\varrho}^{M_{1}},...,\lambda_{n}\tilde{\varrho}_{n}^{M_{n}})\in\Lambda^{M_{1}}\times...\times\Lambda^{M_{n}}$,
then it is not optimal for $(\lambda_{1}\varrho^{M_{1}},...,\lambda_{n}\varrho^{M_{n}})$.}
\item \textit{If $(X_{1},...,X_{n})$ is an optimal allocation for $(\lambda_{1}\varrho^{M_{1}},...,\lambda_{n}\varrho^{M_{n}})$
then there exists $Y\in L^{q}$ such that $\lambda_{i}\varrho^{M_{i}}(X_{i})=E(X_{i}Y),i=1,...,n$.}\end{enumerate}
\begin{rem}
From pricing point of view, in the third statement of the previous
theorem, $Y$ can be interpreted as the ``generalized stochastic discount
factor''. For further reading on the relation between the set of stochastic
discount factors and the optimal risk allocations see \citet{Filipovic/Kupper:2008}.

In the following theorem we study the existence of an asymptotic optimal
allocation.\end{rem}
\begin{thm}
\label{thm:main_general}Let $\varrho_{1},...,\varrho_{n}$ be $n$
distortion risk measures, and for each $i$=1,...,n, let $\Lambda_{i}$
denote the set of all coherent risk measures $\tilde{\varrho_{i}}\ge\varrho_{i}$.
If the total risk $X_{0}$ is bounded below by $M\in\mathbb{R}$,
\eqref{eq:main_minimization-1} is bounded if and only if $\cap_{i}\lambda_{i}\Delta_{\tilde{\varrho_{i}}}\not=\varnothing$
for all $(\tilde{\varrho_{1}},...,\tilde{\varrho}_{n})\in\Lambda_{1}\times....\times\Lambda_{n}$.\end{thm}
\begin{proof}
By Theorem \ref{por:1}
\begin{multline}
\inf_{X_{1}+...+X_{n}=X_{0}}\lambda_{1}\varrho_{1}(X_{1})+...+\lambda_{n}\varrho_{n}(X_{n})\\
=\inf_{X_{1}+...+X_{n}=X_{0}}\bigl\{\min\limits _{\tilde{\varrho_{1}}\in\Lambda_{1}}\lambda_{1}\tilde{\varrho_{1}}(X_{1})+...+\min\limits _{\tilde{\varrho_{n}}\in\Lambda_{n}}\lambda_{n}\tilde{\varrho_{n}}(X_{n})\bigr\}\\
=\inf\limits _{(\tilde{\varrho_{1}},...,\tilde{\varrho}_{n})\in\Lambda_{1}\times....\times\Lambda_{n}}\inf\limits _{X_{1}+...+X_{n}=X}\bigl\{\lambda_{1}\tilde{\varrho}_{1}(X_{1})+...+\lambda_{n}\tilde{\varrho}_{n}(X_{n})\bigr\}\\
=\inf\limits _{(\tilde{\varrho_{1}},...,\tilde{\varrho}_{n})\in\Lambda_{1}\times....\times\Lambda_{n}}\sup\limits _{Y\in\cap_{i}\lambda_{i}\Delta_{\tilde{\varrho_{i}}}}E(YX_{0}).\label{eq:proof_existence_general}
\end{multline}
It is clear that if the infimum in \eqref{eq:proof_existence_general}
is bounded then all intersections $\cap_{i}\lambda_{i}\Delta_{\tilde{\varrho_{i}}}$,
for all $(\tilde{\varrho_{1}},...,\tilde{\varrho}_{n})\in\Lambda_{1}\times....\times\Lambda_{n}$,
are non-empty. On the other hand, since $\forall Y\in\Delta_{\tilde{\varrho}_{i}},i=1,...,n$,
$Y\ge0\text{, }E(Y)=1$ and $X_{0}\ge M$, we have that $E(YX_{0})\ge-\vert M\vert$.
This implies that if all the intersections $\cap_{i}\lambda_{i}\Delta_{\tilde{\varrho_{i}}}$,
for all $(\tilde{\varrho_{1}},...,\tilde{\varrho}_{n})\in\Lambda_{1}\times....\times\Lambda_{n}$,
are non-empty, then the right hand side of \eqref{eq:proof_existence_general}
is bounded below by $-\vert M\vert\max_{i}\lambda_{i}$, and therefore,
\eqref{eq:proof_existence_general} is bounded.
\end{proof}
Now we have the following corollaries
\begin{cor}
The boundedness of the problem \eqref{eq:main_minimization-1} is
independent of the total risk.
\end{cor}

\begin{cor}
For $X_{0}\ge0$, \eqref{eq:main_minimization-1} has a solution if
and only if \textup{$\lambda_{1}=...=\lambda_{n}$ and $\cap_{i}\Delta_{\tilde{\varrho_{i}}}\not=\varnothing$.} \end{cor}
\begin{example}
Let $\varrho_{1}=\mathrm{VaR_{\alpha}}$ and $\varrho_{2}=E$, and
let us assume $X_{0}$ is any arbitrary random loss. According to
Theorem \ref{thm:main_general}, the optimal risk allocation problem
\eqref{eq:main_minimization-1} has a solution if $P\in\Delta_{\tilde{\varrho}}$
for any coherent risk measure $\tilde{\varrho}\ge\mathrm{VaR_{\alpha}}$.
On the other hand, according to Theorem \ref{por:1} for any $X\in L^{p}$,
$\mathrm{VaR}_{\alpha}(X)=\tilde{\varrho}(X)$ for some coherent risk
measure $\tilde{\varrho}\ge\mathrm{VaR_{\alpha}}$. This implies that
$\mathrm{VaR_{\alpha}}(X)\ge E(X)$, for any $X\in L^{p}$. This inequality
clearly does not hold, if we choose $X=1_{A}$ for some set $A\in\mathcal{F}$
that $0<P(A)<\frac{1-\alpha}{2}$.
\end{example}
Theorems 4 and 5 can be considered as generalization of many existing
papers in the literature where their result can only be applied to
coherent risk measures, which in our setting is to use singleton sets$\Lambda_{i}=\left\{ \varrho_{i}\right\} $;
see for instance \citet{Jouini/Schachermayer/Touzi:2008} , \citet{Filipovic/Kupper:2008}
and \citet{Filipovic/Kupper:2008(1)}

\subsection{Co-monotone allocations and Admissible Allocations }

Concerning the discussion we had about moral hazard, in this section
we assume that all contracts in the market are designed to be co-monotone.
To set this economic assumption on a sound mathematical basis, we
assume that all contracts are non-decreasing functions of the total
risk. Therefore, in a market with this assumption any allocation $(X_{1,}...,X_{n})$
is equal to $(f_{1}(X_{0}),...,f_{n}(X_{0}))$ when $f_{1},...,f_{n}$
are $n$ non-negative and non-decreasing functions such that $f_{1}+...+f_{n}=id$. 

We introduce the set of allocations as 
\[
\mathbb{C}=\left\{ f\in L_{+}^{0}(\mathbb{R}_{+})\Big|f\text{ is nondecreasing and }f(0)=0\right\} .
\]
and the set of admissible allocation as

\[
\mathbb{AC}=\left\{ (f_{1},...,f_{n})\in\mathbb{C}^{n}|f_{1}+....+f_{n}=id\right\} .
\]
$\mathbb{AC}$ is a closed, convex and weakly compact set of $L^{p}(\mathbb{R})$
for $p\in[1,\infty)$. On the other hand, it is easy to see that any
component $f_{i}$, is a Lipschitz function of degree one, i.e. $0\le f_{i}(y)-f_{i}(x)\le y-x,\text{ for }0\ensuremath{\le}x\ensuremath{\le}y$.
Indeed, it is enough to check it for $n=2$. In this paper, we focus
our attention to the allocation set induced by $\mathbb{AC}$

\[
\mathbb{AA}=\left\{ (f_{1}(X_{0}),...,f_{n}(X_{0}))|(f_{1},....,f_{n})\in\mathbb{AC}\right\} .
\]
\citet{Filipovic/Svindland:2008(1)} prove that for a set of $n$
law and cash invariant convex functions $\varrho_{1},...,\varrho_{n}$,
any solution $(X_{1},...,X_{n})$ to \eqref{eq:convolution_allocation}
is co-monotone. In particular this means that in a market with convex
distortion risks the optimal allocations are automatically from $\mathbb{{AC}}.$
This is no longer true for the general case as shown in the following
example. 
\begin{example}
Let us assume $\varrho_{1}=\mathrm{VaR}_{\alpha}$, $\varrho_{2}=\mathrm{VaR}_{\beta}$,
$X_{0}>0,\, a.s.$, $\alpha+\beta>1$ and $0<\alpha<\beta<1$. Let
us assume $X_{0}$ is a random variable with a strictly increasing
and continuous CDF function $F_{X_{0}}$. Since $n=2$ in this example,
one can assume that there is a function $f$ such that $f$ and $id-f$
are non-negative, non-decreasing and that $f_{1}=f$ and $f_{2}=id-f$.
We first prove the following lemma \end{example}
\begin{lem}
There is a positive number $c>0$ such that for any function $f$
described above, the following inequality holds 
\[
\mathrm{VaR}_{\alpha}(f(X_{0}))+\mathrm{VaR}_{\beta}(X_{0}-f(X_{0}))>c+\mathrm{VaR}_{\alpha+\beta-1}(X_{0}).
\]
\end{lem}
\begin{proof}
It is known that Value at Risk can commute with a non-decreasing function,
therefore,

\begin{align*}
 & \mathrm{VaR}_{\alpha}(f(X_{0}))=f(\mathrm{VaR}_{\alpha}(X_{0})),\\
 & \mathrm{VaR}_{\beta}(X_{0}-f(X_{0}))=\mathrm{VaR}_{\beta}(X_{0})-f(\mathrm{VaR}_{\beta}(X_{0})).
\end{align*}
Strict monotonicity of $F_{X_{0}}$ , $\alpha<\beta$ and $\alpha+\beta-1<\alpha$
imply
\begin{align*}
\mathrm{VaR}_{\alpha}(f(X_{0}))+\mathrm{VaR}_{\beta}(X_{0}-f(X_{0})) & =f(\mathrm{VaR}_{\alpha}(X_{0}))+\mathrm{VaR}_{\beta}(X_{0})-f(\mathrm{VaR}_{\beta}(X_{0}))\\
 & \ge\mathrm{VaR}_{\beta}(X_{0})+(\mathrm{VaR}_{\alpha}(X_{0})-\mathrm{VaR}_{\beta}(X_{0}))\\
 & =\mathrm{VaR}_{\alpha}(X_{0})>c+\mathrm{VaR}_{\alpha+\beta-1}(X_{0}),
\end{align*}
where $c=\frac{\mathrm{VaR}_{\alpha}(X_{0})-\mathrm{VaR}_{\alpha+\beta-1}(X_{0})}{2}$.
\end{proof}
\vspace{0.5cm}
 The result of the lemma is that there is no \textit{admissible allocation}
which can attain the value $\mathrm{VaR}_{\alpha+\beta-1}(X_{0})$.
Now let us consider the allocation $X_{1}=X_{0}1_{\{X_{0}>\mathrm{VaR}_{\alpha}(X_{0})\}}$.
It is clear that $P(X_{1}>0)=1-\alpha$, meaning that $\mathrm{VaR}_{\alpha}(X_{1})=0$.
On the other hand,

\begin{align*}
P(X_{2}>x) & =P(X_{0}>x\,\&\,\mathrm{VaR}_{\alpha}(X_{0})\ge X_{0})\\
 & =P(X_{0}\le\mathrm{VaR_{\alpha}(X_{0})})-P(X_{0}\le x)\\
 & =\alpha-F_{X_{0}}(x).
\end{align*}
This simply implies that $F_{X_{2}}(x)=1+F_{X_{0}}(x)-\alpha$, and
therefore $\mathrm{VaR}_{\beta}(X_{2})=\mathrm{VaR}_{\alpha+\beta-1}(X_{0})$.
Hence, we have that $\mathrm{VaR}_{\alpha}(X_{1})+\mathrm{VaR}_{\beta}(X_{2})=\mathrm{VaR}_{\alpha+\beta-1}(X_{0})$.

Allocation $(X_{1},X_{2})$ is an example of a moral hazard situation,
where agent 2 is not sensitive to the big total losses. This example
shows why in a market with non-convex beliefs we have to further assume
that there is no risk of moral hazard.
\begin{rem}
Observe that if all agents in the market use the same risk measure
$\varrho$, by using the fact that $\mathrm{VaR}$ commutes with non-decreasing
functions, we have
\end{rem}
\[
\varrho(f_{1}(X_{0}))+...+\varrho(f_{n}(X_{0}))=\varrho(X_{0}),\forall(f_{1},...,f_{n})\in\mathbb{AC}.
\]
This means, no matter what allocation the agents use, as far as there
is no risk of moral hazard, the value of the systemic risk remains
constant. This may happen if the regulator imposes a unique risk measure
to all agents, for example the same $\mathrm{VaR}_{0.995}$ as in
the Solvency II, to measure the capital reserve.

\subsection{Marginal Risk Allocations}

It is known that every Lipschitz continuous function $f$ is almost
everywhere differentiable and its derivative is essentially bounded
by its Lipschitz constant. Furthermore, $f$ can be written as the
integral of its derivative denoted by, i.e., $f(x)=\int_{0}^{x}h(t)dt$.
Therefore, the set $\mathbb{C}$ can be represented as

\[
\mathbb{C}=\left\{ f\in L^{0}(\mathbb{R}_{+})\Big|f(x)=\int_{0}^{x}h(t)dt,0\le h\le1\right\} .
\]
Let us introduce the space of \textit{marginal risk allocations }as

\[
\mathbb{D}=\left\{ h\in L^{0}(\mathbb{R}_{+})\Big|0\le h\le1\right\} .
\]

\begin{defn}
For any function $f\in\mathbb{C}$, the associated marginal risk allocation
is a function $h\in\mathbb{D}$ such that 
\[
f(x)=\int_{0}^{x}h(t)dt,x\ge0.
\]

\end{defn}
The interpretation of marginal risk allocation is as follows: if $f(x)=\int_{0}^{x}h(t)dt$
is in $\mathbb{C}$, then at each value $X_{0}=x$, a marginal change
$\delta$ to the value of the total risk will result in marginal change
of the size $\delta h(x)$ in the allocation risk. We will see in
the following that this marginal change is either $0$ or $\delta$,
i.e., $h=0\text{ or }1$. This means that for any small change in
the total risk, there is only one agent who has to tolerate the changes
in the risk.

\subsection{Co-monotone Optimal Risk Allocations}

Throughout this section we assume $X_{0}\ge0\text{ and }F_{X_{0}}(0)=$0.
Furthermore, we restrict our attention to a family of distortion risk
measures which satisfy the following regularity condition

\begin{equation}
\lim_{m\to\infty}\varrho_{i}(X\wedge m)=\varrho_{i}(X),i=1,...,n.\label{eq:continuity}
\end{equation}
Let $\Psi(t)=\min\left\{ \lambda_{1}(1-\Phi_{1}(t)),...,\lambda_{n}(1-\Phi_{n}(t))\right\} $.
Suppose $k_{i}^{*},i=1,...,n$ is a set of functions that

\begin{equation}
k_{i}^{*}(t)=\left\{ \begin{array}{lll}
1, & \text{ if }\lambda_{i}\left(1-\Phi_{i}(t)\right)<\lambda_{j}\left(1-\Phi_{j}(t)\right)\forall i\not=j\\
0, & \text{ if }\lambda_{i}\left(1-\Phi_{i}(t)\right)>\lambda_{j}\left(1-\Phi_{j}(t)\right)\exists i\not=j
\end{array}\right.,\label{eq:h*}
\end{equation}
where also $k_{1}^{*}+...+k_{n}^{*}=1$. Here we state the main result
of this section
\begin{thm}
\label{thm:main-bang-bang}If $\varrho_{1},...,\varrho_{n}$ satisfy
\eqref{eq:continuity}, the co-monotone solutions to the optimization
problem \eqref{eq:main_minimization-1} is given by $X_{i}=f_{i}^{*}\left(X_{0}\right)$
when 
\begin{equation}
f_{i}^{*}(x)=\int_{0}^{x}k_{i}^{*}(\mathrm{VaR}_{t}(X_{0}))dt,i=1,...,n.\label{eq:separate}
\end{equation}
 Furthermore, the value at minimum is given by 
\begin{equation}
\int_{0}^{\infty}\Psi(s)ds.\label{EQ:MINIMUM}
\end{equation}
\end{thm}
\begin{proof}
Let $\varrho_{i}=\int_{0}^{1}\mathrm{VaR}_{t}(X_{0})d\Phi_{i}(t)\,,\, i=1,...,n$.
Then for any member $(f_{1},....,f_{n})$ from the set $\mathbb{AC}$,
using the fact that VaR always commutes with non-decreasing functions,
we have
\begin{multline}
\lambda_{1}\varrho_{1}(f_{1}(X_{0}))+....+\lambda_{n}\varrho_{n}(f_{n}(X_{0}))\\
=\int_{0}^{1}\lambda_{1}\mathrm{VaR}_{t}(f_{1}(X_{0}))d\Phi_{1}(t)+...+\int_{0}^{1}\lambda_{n}\mathrm{VaR}_{t}(f_{n}(X_{0}))d\Phi_{n}(t)\\
=\int_{0}^{1}\lambda_{1}f_{1}(\mathrm{VaR}_{t}(X_{0}))d\Phi_{1}(t)+...+\int_{0}^{1}\lambda_{n}f_{n}(\mathrm{VaR}_{t}(X_{0}))d\Phi_{n}(t).\label{eq:eq_tum_1}
\end{multline}
Let us denote the derivatives of $f_{1},...,f_{n}$ by $h_{1},....,h_{n}$.
Therefore,
\begin{multline*}
\lambda_{1}\varrho_{1}(f_{1}(X_{0}))+....+\lambda_{n}\varrho_{n}(f_{n}(X_{0}))\\
=\int_{0}^{1}\left(\int_{0}^{\mathrm{{VaR}}_{t}(X_{0})}\lambda_{1}h_{1}(s)ds\right)d\Phi_{1}(t)+...+\int_{0}^{1}\left(\int_{0}^{\mathrm{{VaR}}_{t}(X_{0})}\lambda_{n}h_{n}(s)ds\right)d\Phi_{n}(t).
\end{multline*}
First, we assume $X_{0}$ is bounded. By Fubini's Theorem we have
\begin{multline}
\lambda_{1}\varrho_{1}(f_{1}(X_{0}))+....+\lambda_{n}\varrho_{n}(f_{n}(X_{0}))\\
=\int_{0}^{\infty}\left[\left(\int_{F_{X_{0}}(s)}^{1}\lambda_{1}d\Phi_{1}(t)\right)h_{1}(s)+...+\left(\int_{F_{X_{0}}(s)}^{1}\lambda_{n}d\Phi_{n}(t)\right)h_{n}(s)\right]ds\\
=\int_{0}^{\infty}\left[\lambda_{1}\left(1-\Phi_{1}\left(F_{X_{0}}(s)\right)\right)h_{1}(s)+...+\lambda_{n}\left(1-\Phi_{n}\left(F_{X_{0}}(s)\right)\right)h_{n}(s)\right]ds\\
\label{eq:eq_tum_3}
\end{multline}
where we use the fact that $\Phi_{1}(1)=...=\Phi_{n}(1)=1$. It is
now clear that the following $(h_{1}^{*},...,h_{n}^{*})$ will minimize
\eqref{eq:eq_tum_3}
\begin{equation}
h_{i}^{*}(s)=\left\{ \begin{array}{lll}
1, & \text{ if }\lambda_{i}(1-\Phi_{i}(F_{X_{0}}(s)))<\lambda_{j}(1-\Phi_{j}(F_{X_{0}}(s))),\forall i\not=j\\
0 & \text{ if }\lambda_{i}(1-\Phi_{i}(F_{X_{0}}(s)))>\lambda_{j}(1-\Phi_{j}(F_{X_{0}}(s))),\exists i\not=j
\end{array}\right.
\end{equation}
where also $h_{1}^{*}+...+h_{n}^{*}=1$. The value of the minimum
also is equal to 
\begin{equation}
\int_{0}^{\infty}\Psi(s)ds.
\end{equation}
If we make a simple change of variable $t=F_{X_{0}}(s)$, we get the
result.

Now assume the general case when $X_{0}$ is not bounded. It is clear
that at each point $t$, for every $i$ between 1 and $n$, $\left\{ \Phi_{i}\circ F_{X_{0}\wedge m}(t)\right\} _{m=1}^{\infty}$
is non-increasing with respect to $m$. On the other hand, for any
$t$, there exist $m_{t}$ such that if $m>m_{t}$ then $F_{X_{0}\wedge m}(t)=F_{X_{0}}(t)$.
Therefore, at each point $t$, we have that $\Phi_{i}(F_{X_{0}\wedge m}(t))\downarrow\Phi_{i}(F_{X_{0}}(t))$.
By monotone convergence theorem we have that 
\[
\lim\limits _{m\to\infty}\int\limits _{0}^{\infty}\Phi_{i}(F_{X_{0}\wedge m}(t))h(t)dt=\int\limits _{0}^{\infty}\Phi_{i}(F_{X_{0}}(t))h(t)dt,
\]
for any function $h\in\mathbb{D}$. Using this fact and our continuity
assumption
\begin{align*}
\varrho_{i}(f(X_{0})) & =\lim_{m\to\infty}\varrho_{i}(f(X_{0})\wedge f(m))\\
 & =\lim_{m\to\infty}\varrho_{i}(f(X_{0}\wedge m))\\
 & =\lim\limits _{m\to\infty}\int_{0}^{\infty}(1-\Phi_{i}(F_{X_{0}\wedge m}(s)))h(s)ds\\
 & =\int_{0}^{\infty}(1-\Phi_{i}(F_{X_{0}}(s)))h(s)ds
\end{align*}
This simply results in
\begin{multline*}
\lambda_{1}\varrho_{1}(f_{1}(X_{0}))+....+\lambda_{n}\varrho_{n}(f_{n}(X_{0}))\\
=\int_{0}^{\infty}[\lambda_{1}(1-\Phi_{1}(F_{X_{0}}(s)))h_{1}(s)+...+\lambda_{n}(1-\Phi_{n}(F_{X_{0}}(s))))h_{n}(s)]ds
\end{multline*}
The rest of the proof follows the same lines after \eqref{eq:eq_tum_3}.\end{proof}
\begin{rem}
As one can see form the last theorem, $k_{i}^{*},i=1,...,n$ only
depend on market preferences, and therefore, they are universal. Also
one can see from \eqref{eq:separate} how the role of the total risk
and the market preferences are separated.
\end{rem}

\begin{rem}
In Theorem \ref{thm:main-bang-bang}, it is shown that the marginal
risk allocations take only the values zero or one. There are some
similar results in the literature of actuarial mathematics which can
prove this for very particular settings, for instance; see \citet{Cai/Tan/Weng/Zhang:2008},
\citet{Cheung:2010},\citet{Chi:2012}, \citet{Chi:2012(2)}, \citet{Chi/Tan:2013}
, \citet{Cheung/Sung/Yam/Yung:2014} and more recently \citet{Assa:2015}.
Theorem \ref{thm:main-bang-bang} can extend all those works form
two different aspects. First, we use a larger family of risk measures
and premiums (distortion risk measures and premiums) which include
almost all risk measures such as VaR and CVaR and risk premiums such
as Wang's premium, used by them. Second, our work can increase the
number of players from two (insurance and re-insurance company) to
$n$, which otherwise, by using the techniques from the existing literature
would be either impossible or at least very difficult to do.\end{rem}
\begin{cor}
If $\lambda_{1}=....=\lambda_{n}$ then $\varrho_{1}\boxdot...\boxdot\varrho_{n}=\varrho_{\Phi}$
when $\Phi=\max\left\{ \Phi_{1},....,\Phi_{n}\right\} $.\end{cor}
\begin{example}
Let us consider the example we discussed earlier. Let us consider
that there are two companies using $\mathrm{\varrho_{1}=VaR}_{\alpha}$
and $\varrho_{2}=\mathrm{VaR}_{\beta}$, where $\alpha<\beta$. It
is clear that since $\alpha<\beta$ one solution is $h_{1}=1$ and
$h_{2}=0$ and $\varrho_{1}\square\varrho_{2}(X_{0})=\mathrm{VaR_{\alpha}(X_{0})}$.\textit{}\\
\textit{}\\
\textit{Conflict of Interest: }The author declares that he has no
conflict of interest.
\end{example}

\end{document}